\documentclass[11pt,a4paper]{article}
\usepackage{fullpage}
\usepackage{times}
\usepackage{amssymb,amsmath}
\usepackage{algorithm}
\usepackage[noend]{algpseudocode}

\newtheorem{theorem}{Theorem}[section]
\newtheorem{lemma}[theorem]{Lemma}

\newtheorem{corollary}[theorem]{Corollary}

\newtheorem{definition}[theorem]{Definition}

\newcommand{\sq}{\hbox{\rlap{$\sqcap$}$\sqcup$}}
\newcommand{\qed}{\hspace*{\fill}\sq}
\newenvironment{proof}{\noindent {\bf Proof.}\ }{\qed\par\vskip 4mm\par}

\begin{document}
 
\date{}
\begin{title}
{A Deterministic Worst-Case Message Complexity Optimal Solution for Resource Discovery\thanks{This work was partially supported by the German Research Foundation (DFG) within the
Collaborative Research Centre �On-The-Fly Computing� (SFB 901).}}
\end{title}

\author{Sebastian Kniesburges \\
   University of Paderborn \\
   seppel@upb.de \\
   \and
   Andreas Koutsopoulos \\
   University of Paderborn  \\
   koutsopo@mail.upb.de\\
   \and
   Christian Scheideler\\
   University of Paderborn  \\
   scheideler@mail.upb.de
  }


\maketitle 


\begin{abstract}
We consider the problem of resource discovery in distributed systems. In particular we give an algorithm, such that each node in a network discovers the address of any other node in the network. We model the knowledge of the nodes as a virtual overlay network given by a directed graph such that complete knowledge of all nodes corresponds to a complete graph in the overlay network. Although there are several solutions for resource discovery, our solution is the first that achieves worst-case optimal work for each node, i.e. the number of addresses ($\mathcal O(n)$) or bits ($\mathcal O(n\log n)$) a node receives or sends coincides with the lower bound, while ensuring only a linear runtime ($\mathcal O(n)$) on the number of rounds.
\end{abstract}




\section{Introduction}

To perform cooperative tasks in distributed systems the network nodes have to know which other nodes are participating. Examples for such cooperative tasks range from fundamental problems such as group-based cryptography \cite{MSU08}, verifiable secret sharing \cite{CGMA85}, distributed consensus \cite{PSL80}, and broadcasting \cite{RC05}
 to peer-to-peer(P2P) applications like distributed storage, multiplayer online gaming, and various social network applications such as chat groups. To perform these tasks efficiently knowledge of the complete network for each node is assumed. Considering large-scale, real-world networks this complete knowledge has to be maintained despite high dynamics, such as joining or leaving nodes, that lead to changing topologies. Therefore the nodes in a network need to learn about all other nodes currently in the network.
This problem called \emph{resource discovery}, i.e. the discovery of the addresses of all nodes in the network by every single node, is a well studied problem and was firstly introduced by Harchol-Balter, Leighton and Lewin in \cite{HBLL99:discovery}. 

\subsection{Resource Discovery}
As mentioned in \cite{HBLL99:discovery} the resource discovery problem can be solved by a simple swamping algorithm also known as {\em pointer doubling}: in each round, every node informs all of its neighbors about its entire neighborhood. While this just needs $O(\log n)$ communication rounds to inform every node about any other node in every weakly connected network of size $n$,
 the work spent by the nodes can be very high and far from optimal. 
We measure the work of a node as the number of addresses each node receives or sends while executing the algorithm. Moreover, in the stable state (i.e., each node has complete knowledge) the work spent by every node in a single round is $\Theta(n^2)$, which is certainly not useful for large-scale systems. Alternatively, each node may just introduce a single neighbor to all of its neighbors in a round-robin fashion. However, it is easy to construct initial situations in which this strategy is not better than pointer doubling in order to reach complete knowledge. The problem in both approaches is the high amount of redundancy: addresses of nodes may be sent to other nodes that are already aware of that address.
In \cite{HBLL99:discovery} a randomized algorithm called the \emph{Name-Dropper} is presented that solves the resource discovery problem within $\mathcal O(\log^2 n)$ rounds w.h.p. and work of $\mathcal O(n^2 \log^2 n)$. 
In \cite{KPV01:det-disc} a deterministic solution for resource discovery in distributed networks was proposed by Kutten et al. Their solution uses the same model as in \cite{HBLL99:discovery} and improves the number of communication rounds to 
which takes $\mathcal O(\log n)$ rounds and $\mathcal O(n^2\log n)$ amount of work. Konwar et al. presented solutions for the resource discovery problem considering different models, i.e. multicast or unicast abilities and messages of different sizes, where the upper bound for the work is $O(n^2\log^2 n) $. In their algorithms they also considered when to terminate, i.e. how can a node detect that its knowledge is already complete. 
Recently resource discovery has been studied by Haeupler et. al. in \cite{Raj}, in which they present two simple randomized algorithms based on gossiping that need $\Omega (n\log n)$ time and $\Omega (n^2 \log n)$ work per node on expectation. They only allow nodes to send a single message containing at most one address of size $\log n$ in each round. Thus their model is more restrictive compared to the model used in \cite{HBLL99:discovery,KPV01:det-disc} and leads to an increased runtime in the number of rounds. 
We present a deterministic solution that follows the idea of \cite{Raj} and limits the number of messages each node has to send and the number of addresses transmitted in one message. Our goal is to reduce the number of messages sent and received by each node such that we avoid nodes to be overloaded. In detail we show that resource discovery can be solved in $\mathcal O(n)$ rounds and it suffices that each node sends and receives $\mathcal O(n)$ messages in total, each message containing $\mathcal O(1)$ addresses. Our solution is the first solution for resource discovery that not only considers the total number of messages but also the number of messages a single node has to send or receive.
Note that $\Omega (n)$ is a trivial lower bound for the work of each node to gain complete knowledge: starting with a list, in which each node is only connected to two other nodes, each node has to receive at least $n-3$ IDs. 
So our algorithm is worst case optimal in terms of message complexity. Furthermore our algorithm can handle the deletion of edges and joining or leaving nodes, as long as the graph remains weakly connected. 
Modeling the current knowledge of all nodes as a directed graph, i.e. there is an edge $(u,v)$ iff $u$ knows $v$'s ID, one can think of resource discovery as building and maintaining a complete graph, a clique, as a virtual overlay network. If the overlay can be recovered out of any (weakly connected) initial graph, the corresponding algorithm can be considered to be a \emph{self-stabilizing} algorithm. More precisely, an algorithm is considered as self-stabilizing if it reaches a legal state when started in an arbitrary initial state (\emph{convergence}) and stays in a legal state when started in a legal state (\emph{closure}). 

\subsection{Topological Self-Stabilization}
There is a large body of literature on how to efficiently maintain overlay
networks, e.g., \cite{AS03:skip,AS04:hyperring,BK+04:Pagoda,
RD01:pastry,HJ+03:skipnet,KSW05:self-repair,
MNR02:viceroy,NW07:continous-discrete,RF+01:can,SM+01:chord,SS09:oblivious-heap}. While many results are already known on how
to keep an overlay network in a legal state, far less is known about
self-stabilizing overlay networks. The idea of self-stabilization in
distributed computing first appeared in a classical paper by E.W. Dijkstra in
1974 \cite{D74:self-stab} in which he looked at the problem of
self-stabilization in a token ring. Interestingly, though self-stabilizing
distributed computing has received a lot of attention for many years, the
problem of designing self-stabilizing networks has attracted much less attention. 
In order to recover certain network topologies from any weakly connected network, researchers have started with simple line and ring networks, \cite{CF05:stab-ring,SR05:ring}.
The Iterative Successor Pointer Rewiring Protocol \cite{CF05:stab-ring} and the Ring Network \cite{SR05:ring}, for example, organize the nodes in a sorted ring. 
In \cite{DK08} Dolev and Kat describe a strategy to  build a hypertree with a polylogarithmic degree and search time. 
In \cite{ORS07:lin}, Onus et al. present a local-control strategy called linearization for converting an arbitrary connected graph into a sorted list.
Various self-stabilzing algorithms for different network overlay structures have been considered over the years  \cite{JRS+09:delaunay,JR+09:skip+,DT09,DT10,DK08}.
Jacob et al. \cite{JRS+09:delaunay} generalize insights gained from graph linearization to two dimensions and present a self-stabilizing construction for Delaunay graphs. In another paper, Jacob et al. \cite{JR+09:skip+} present a self-stabilizing variant of the skip graph and show that it can recover its network topology from any weakly connected state in $\mathcal O(\log^2 n)$ communication rounds with high probability. In \cite{DT09} and \cite{DT10} Dolev and Tzachar show self-stabilizing algorithms for forming subgraphs like clusters  or expanders in just polylogarithmic number of rounds. 
In \cite{DT09} the authors use a self-stabilizing algorithm in which they collect snapshots of the network along a spanning tree, which could also be used to form a complete graph. However, the authors give no bounds on the message complexity of their algorithm.
In \cite{BGP10:frame} the authors present a general
framework for the self-stabilizing construction of overlay networks, which may involves the construction of the clique.
The algorithm requires the knowledge of the 2-hop neighborhood for each node and
may involve the construction of a clique. In that way, failures at the
structure of the overlay network can easily be detected and repaired. 
However, the work in order to do that when using this method is too high as they essentially use pointer doubling, i.e. in each round a node sends the information about its neighborhood to all its neighbors.

One could use the distributed algorithms for self-stabilizing lists and rings to form a complete graph, but all algorithms proposed so far for these topologies involve a worst-case work of $\Omega(n^2)$ per node in order to form the list or ring. Hence, these algorithms cannot be used to obtain an efficient algorithm for the clique.

Alternatively, a self-stabilizing spanning tree algorithm could be used. A
large number of self-stabilizing distributed algorithms has already been
proposed for the formation of spanning trees in static network topologies, \cite{BPR11}, \cite{BDPR10}, \cite{HLPPB07}, \cite{HLPPB07}.
For example in \cite{BPR11} the authors present a self-stabilizing spanning tree with minimal degree for the given network and in \cite{BDPR10} a fast algorithm for a self-stabilizing spanning tree is presented, which reaches optimal convergence time $\mathcal O(n^2)$ in an asynchronous setting.
However, these spanning trees are either expensive to maintain or the amount of work in these algorithms is not being considered.

However, these spanning trees are potentially expensive to maintain as a high
degree cannot be avoided in general (consider, for example, the extreme case
of a star graph in which a single node is connected to all other nodes). For
the case that the network topology is flexible and potentially allows every
node to connect to any other node, self-stabilizing algorithms are known that
construct a bounded degree spanning tree (e.g., \cite{HLPPB07}). The algorithm
in \cite{HLPPB07} also has a very low overhead in the stable state. But no
formal result is given on the work to establish the spanning tree. Also, an
outside rendezvous service, called an oracle, is used to introduce nodes to other nodes, which is not available in our model.

In summary, no self-stabilizing algorithm has been presented for the formation
of a bounded degree spanning tree if the network topology is under the control
of the nodes and there are no outside services for the introduction of nodes.

\subsection{Our model}
We use the network model used in \cite{HBLL99:discovery,KPV01:det-disc,Raj}. In the following we give a detailed description of the model. 
We model the network as a directed graph $G = (V,E)$ where $|V| = n$. The
nodes have unique identifiers with a total order, and these identifiers are
assumed to be immutable (for example, we may use the IP addresses of the
nodes). We are using a standard synchronous message-passing model: time proceeds in synchronous rounds, and all messages generated in round $i$ are delivered at the end of round $i$. In order to deliver a message, a node may use any address stored in its local variables. In each
round, each node can only inspect its local variables (i.e. it can only
communicate with nodes that it knows). Beyond that, a node does not have
access to any information or services which means, for example, that 
No a priori information about the size or diameter of the network can be assumed by a node
and there cannot be made use of some outside rendezvous service to get
introduced to other nodes. Hence, the {\em state} of a node is fully
determined by its local variables. Like in \cite{HBLL99:discovery,KPV01:det-disc,Raj} we assume that a node can verify its neighborhood without extra work, such that there are no false identifiers in the network. 
Only local topology changes are allowed, i.e. a node may decide to cut a link
to a neighbor (by deleting its address) or introduce a link to one of its
neighbors (by sending it an address). We model the decisions to cut or establish links and to send messages as actions. An action has the form $<guard>\rightarrow <commands>$.
A guard is a Boolean expression over the state of the node. The commands are
executed if the guard is true. Any action whose guard is true is said to be
{\em enabled}. We assume that a node can execute all of its enabled actions in the current round.

The {\em state} of the system is the combination of the states of all nodes in
the system. Due to our synchronous message-passing model, in which no message
is still in transit at the beginning of a round, the state of the system
and contains all the information available in the system. A {\em computation} is a
sequence of system states such that for each state $s_i$ at the beginning of
round $i$, the next state $s_{i+1}$ is obtained after executing all actions
that are enabled at the beginning of round $i$ and receiving all messages that
they generated. We call a distributed algorithm {\em self-stabilizing} if from
any initial state in which the overlay network is weakly connected, it
eventually reaches a legal state and stays in a legal state afterwards.
In our case, the legal state is the clique topology. Since the clique topology is uniquely defined, no more topological changes will happen afterwards.Our goal is to develop algorithms that need as few communication rounds and as little work as possible to arrive at a clique. 
We distinguish between two types of work. The {\em stabilization work} of a node $v$ is defined as the
total number of addresses sent and received by $v$ during the stabilization
process. The {\em maintenance work} of a node $v$ is defined as the maximum
number of addresses sent and received by $v$ during a single round of the
stable state, i.e. for the case that a clique has been formed.

\subsection{Our contributions}

In this paper we present a distributed algorithm for resource discovery. We will describe the algorithm as a self-stabilizing algorithm that forms and maintains a clique as a virtual overlay network. In particular, the following theorem shows that our algorithm is worst-case optimal in terms of message complexity.

\begin{theorem}
For any initial state in which the network is weakly connected, our algorithm
requires at most $\mathcal O(n)$ rounds and $\mathcal O(n)$ work per node
until the network reaches a legal state in which it forms a clique.
\end{theorem}

We further show that the maintenance cost per round is $\mathcal O(1)$ for
each node once a legal state has been reached. We also consider topology
updates caused by a single joining or leaving node and show that the network
recovers in $\mathcal O(n)$ rounds with at most $\mathcal O(n)$ messages over
all nodes besides the maintenance work. Note that we use a synchronous message passing model to give bounds on the message complexity of our algorithm, but our correctness analysis can also be applied to an asynchronous setting.

\subsection{Structure of the paper}

The paper is structured as follows: In Section 2 we give a description of our algorithm. In Section 3 we prove that the algorithm is self-stabilizing. We consider the stabilization work and maintenance work in Section 4. In Section 5 we analyze the steps needed for the network to recover after a node joins or leaves the network. Finally, in Section 6 we end with a conclusion.


\section{A distributed self-stabilizing algorithm for the clique}


In this section we give a general description of our algorithm. First we introduce the variables being used, and then the actions the nodes take, according to our rules.
Each node $x$ has a buffer $B(x)$ for incoming messages from the previous round. We assume that the buffer capacity
is unbounded and no messages are lost. We do not require any particular order
in which the messages are processed in $B(x)$. 
Moreover, each node $x$ stores the following internal variables: its predecessor $p(x)$ , its successor $s(x)$, its current neighborhood $N(x)$ in a circular list, the nodes received by messages from the predecessor in another circular list $L(x)$, the set of nodes $S(x)$ that are received through scanning messages (defined below), its own identifier $id(x)$ and its status $status(x)$, which is by default set to 'inactive' and can be changed to 'active'. The current network $G=(V,E)$ formed by the nodes is defined by their current neighborhoods $N(v)$. We only require that $N(v)$ does not contain false ids, since in that case the stabilization time could be delayed.

A message in general consists of the following parts: a \emph{sender id}, which is the id of the node sending the message, an optional \emph{additional id}, if the sender wants to inform the receiving node about another node, and  the \emph{type} of the message.

Each node has two different kinds of actions that we call \emph{receive} actions
and \emph{periodic} actions. A receive action is enabled if there is an incoming message of the corresponding type in the buffer $B(x)$. There are the
following types of messages: \emph{pred-request, pred-accept, new-predecessor,
deactivate, activate, forward-from-successor, forward-from-predecessor,
forward-head, scan, scanack, delete-successor}.
A periodic action is enabled in
every state, as its guard is simply \emph{true}. Therefore there can be no state in the computation in which no action is enabled. 
Each enabled action is executed once every step.

\subsection{Definitions}
In order to describe the algorithm formally and prove its correctness later on, we need the definitions given below. In this paper we assume that a predecessor of a node is a node with the next larger identifier. Therefore for all $p(x)$ links, $p(x)>x$. Then all nodes in a connected component considering only $p(x)$ links form a rooted tree, where for each tree the root has the largest identifier. Note here that the \emph{heap $H$}  (defined below) is not a data structure or variable stored by any node. It is a notion used just for the purpose of the analysis.

\begin{definition}
We call such a rooted tree formed by $p(x)$ links a \emph{heap $H$}. We further call the root of the tree the \emph{head $h$} of the heap $H$. We further denote with \emph{$heap(x)$} the heap $H$ such that $x\in H$.
\end{definition}

\begin{definition}
A \emph{sorted list} is a heap $H$ with head $h$, such that $\forall v\in H-\left\{h\right\}: p(v)>v$ and $\forall v\in H-\left\{h\right\}: s(p(v))=v$.
We call a heap \emph{linearized w.r.t. a node $u\in H$}, if $\forall v\in H-\left\{h\right\}: p(v)>v$ and $\forall v\in H-\left\{h\right\}\wedge v\geq u: s(p(v))=v$. We further call the time until a heap is linearized w.r.t. a node $u$ the \emph{linearization time of $u$}.
We say that two heaps $H_i$ and $H_j$ are \emph{merged} if all nodes in $H_i$ and $H_j$ form one heap $H$.
\end{definition}

\subsection{Description of our algorithm}
We only present the intuition behind our algorithm. The full pseudocode is in Appendix ~\ref{app:alg}. Our primary goal is to collect the addresses of all nodes in the system at the node of maximum id,
which we also call the {\em root}. In order to efficiently distribute the
addresses from this root to all other nodes in the system (so that all ids are known to every node and a clique is formed), 
we aim at organizing them into a spanning tree of constant degree, which in our case is
a sorted list, ordered in descending ids. The root would
then be the head of the list. In order to reach a sorted list, we first
organize the nodes in rooted trees satisfying the max-heap property, i.e. a
parent (also called {\em predecessor} in the following) of a node has a higher
id than the node itself. The rooted trees will then be merged and linearized
over time so that they ultimately form a single sorted list.

Since we want to minimize our message complexity, we had to look for a
technique other than the linearization technique presented in
\cite{ORS07:lin}. So in our protocol, in order to minimize the amount of
messages sent by the nodes, we allow a node in each round to share information
only with its immediate \emph{successor} $s(x)$ (which is one of the nodes
that considers it as its predecessor) and \emph{predecessor} $p(x)$.
More precisely, in each round a node forwards one of its neighbors (i.e. the
nodes it knows about) in a round-robin manner to its predecessor. The
intuition behind this is that if every node does that sufficiently often,
eventually the root will learn about all ids in the system and will forward
this information in a round-robin manner to its successor, who will then
forward it to its successor, and so on.

In order for this process to work, each node must repeatedly compute and
update its successor and predecessor. This is done as follows: 
Each node chooses the smallest node in its neighborhood that is larger than itself as its predecessor and requests from it to accept it as successor ($pred-request$
message). Each node also looks at the nodes which requested to be its
successor, assigns the largest of them as its successor ($pred-accept$) and
forwards the rest to it ($new-predecessor$). In that way each node has
at most one predecessor and one successor at the end of one round.

We also need to ensure that there exists a path of successors from the root to
all other nodes so that the information can be forwarded to all. This is
initially not the case since there exist many nodes that are the largest in
their known neighborhood, thinking they are the root. We call these nodes
$heads$. All the nodes having the same head as an ancestor form a $heap$. The
challenge is to $merge$ all heaps into one, since then we have only one head, the root.
In order to enable the merging of the heaps, the heads continuously scan their
neighborhood.
A node that receives a $scan$ message responds by sending the largest node in
its neighborhood through a $scanack$ message to the node that sent that $scan$
message (could be possibly more than one). Moreover, in each round, the
largest node is also forwarded to its predecessor ($forward-head$), which in
turn forwards it again to its predecessor, and so on.

We further discuss the process of forwarding an id to a node's
predecessor/successor. Note that when a node forwards an id through a
$forward-from-successor$ resp. $forward-from-predecessor$, the id sent is the one
at the head of the list $N(x)$ resp. $L(x)$. Then the head shifts to the next
element of the (circular) list. When a node receives an id through a
$forward-from-successor$ resp. $forward-from-predecessor$ message, it stores it at
the head of its list. That way we ensure that once a node is forwarded it will
not be delayed by other nodes being forwarded on its way to the root or the
head of the heap. When a node is inserted into a list, the $insert$ operation
is used. The $insert(<list>,<node>,<place>)$ operation works as follows. It
checks whether $<node>$ is already in $<list>$ and if not, it is inserted at
$<place>$, where $<place>$ can be either head or tail (by head here the head
of the list is meant, not the head of a heap as defined above).

To avoid accumulation of unsent ids in the lists (which would have an effect
on the time and message complexity) maintained by the nodes, the following
rules are used. When $x$ has no predecessor that it can send a
\emph{forward-from-successor} message to, although it has neighbors greater than itself (so $x$ is not a head), it changes its status to $inactive$, and then informs its successor through a \emph{deactivate} message in order for $s(x)$ not to send its \emph{forward-from-successor} to $x$, until $x$ has a
predecessor (in that case an \emph{active} message is sent to $s(x)$) to which
it can forward the message. $s(x)$ then changes its status to $inactive$ and
forwards the \emph{deactivate} message to its successor $s(s(x))$, and so on.
In that way no messages that are forwarded to $x$ accumulate at $N(x)$ before
being forwarded again and we ensure that once a node is forwarded, it will not
be delayed by other nodes being forwarded. When $x$ obtains a predecessor, it
will change its status to $active$ and inform through a message of type
\emph{activate} $s(x)$ about that and the information flow can start again.

In order to repair faulty configurations, where a node is thought to be a
successor of more than one node, we introduce the following rule. If a node receives messages sent
by a node that is not its predecessor although the sending node should be the
predecessor, then a node will send a \emph{delete-successor} message,
correcting the wrong $s(x)$ link.


\subsection{Pseudo code}\label{app:alg}
In this last section we will present the pseudo code for the described and analyzed algorithm on the next page. The pseudo code starts with the periodic actions and then shows the receive actions, in which every incoming message is handled according to the specific message type. 
\newline

\begin{algorithm*}[ht]
\tiny
\caption{\textsc{Actions of node x at each round}}
\label{pseudocode:actions}
\begin{algorithmic}

\State \textbf{forwardtopred: true}$\rightarrow $
\If{$status(x) \neq inactive \wedge p(x) \neq null $}                   \Comment{$x$ is not a head}
\State send message(id(x), N(x).head, forward-from-successor) to p(x)         \Comment{forward node to predecessor}
\State N(x).head:=(N(x).head).next                              \Comment{shift head to next element in circular list}
\EndIf
\State

\State \textbf{checkifhead: true}$\rightarrow $
\If{$p(x)=null \vee p(x)<x$}                                    \Comment{$x$ is a head or $p(x)$ is invalid}
\State $p(x):=min\{v \in N(x): v>x\}$
\If{$p(x) \neq null$}
\State send message(id(x),pred-request) to p(x)
\State status(x):=inactive
\Else                               \Comment{$x$ is a head, scan a node}
\State send message(id(x),scan) to N(x).head
\State insert(L(x),N(x).head,tail)              \Comment{a copy of $N(x).head$ is inserted at the end of $L(x)$}
\State N(x).head:=(N(x).head).next
\EndIf
\Else
\State send message(id(x),pred-request) to p(x)
\EndIf

\State

\State \textbf{forwardtosuc:  true}$\rightarrow $
\If{$s(x) \neq null$}
\If {$s(x)<x$}                                          \Comment{test if $s(x)$ is valid}
\State send message(id(x), L(x).head, forward-from-predecessor) to s(x)                          \Comment{forward node to successor}
\State L(x).head:=(L(x).head).next
\Else
\State s(x)=null
\EndIf
\EndIf
\State

\State \textbf{forwardmax: true}$\rightarrow $
\If{$ S(x) \neq null$}
\State $maxN:=\max\{u:u\in N(x)\}$
\State $N(x):=N(x)\cup S(x)$
\State $maxS:=\max\{u:u\in S(x)\}$
\If{$maxS>maxN\wedge p(x)\neq null$}            \Comment{forward largest node}
\State send message(id(x),maxS , forward-head) to p(x)
\State $S:=S\setminus \{maxS\}$
\State $maxN=maxS$
\EndIf
\For{all $u \in S(x)$}                      \Comment{send the maximum to the nodes of $S(x)$}
\State  send message(id(x),maxN,scanack) to u
\State delete(S(x),u)
\EndFor
\EndIf
\State

%

\State \textbf{process: message $m \in B(x) \rightarrow$}
\State

\If{$m.type=forward-head$}                       \Comment{insert the head forwarded from $s(x)$ to $N(x),S(x)$}
\If{$m.id=s(x)$}
\If{$m.id \not\in N(x)$}
\State insert(S(x),m.id)
\EndIf
\State insert(N(x),m.id,tail)
\EndIf
\EndIf
\State

\If{$m.type=scan$}                          \Comment{$x$ has been scanned by a head $m.id$}
\State insert(S(x),m.id)
\EndIf
\State

\If{$m.type=scanack$}
\If{$m.id \not\in N(x)$}
\State insert(S(x),m.id)
\EndIf
\EndIf
\State

\If{$m.type=delete-successor$}
\If{$m.id = s(x)$}
\State $s(x)=null$
\EndIf
\EndIf

\end{algorithmic}
\end{algorithm*}
\begin{algorithm*}[ht]
\tiny
\begin{algorithmic}

\If{$m.type=pred-request$}
\If{$m.id<id(x)$}
\If{$s(x) \neq null$}                               \Comment{renew successor if necessary, and rearrange old successor}
\State grandson:=$\min\{m.id,s(x)\}$
\State s(x):=$\max\{m.id,s(x)\}$
\State send message(id(x),pred-accept)  to s(x)
\State send message(id(x),s(x),new-predecessor) to grandson
\Else
\State s(x):=$m.id$
\State send message(id(x),pred-accept)  to s(x)
\EndIf
\EndIf
\EndIf
\State

\If{$m.type=new-predecessor$}                                            \Comment{renew predecessor}
\If{$m.id=p(x)$}
\If{$m.id2>x \wedge m.id2<p(x)$}
\State p(x)=m.id2
\State send message(id(x),pred-request) to p(x)
\State status(x)=inactive
\If{$s(x)\neq null $}
\State send message(id(x),deactivate) to s(x)
\EndIf
\EndIf
\EndIf
\EndIf
\State

\If{$m.type=pred-accept$}                            \Comment{the predecessor has accepted $x$ as its successor}
\If{$m.id=p(x)$}
\State status(x)=active
\If{$s(x)\neq null $}
\State send message(id(x),activate) to s(x)
\EndIf

\Else
\State send message(id(x),delete-successor) to m.id
\EndIf
\EndIf
\State

\If{$m.type=deactivate$}
\If{$m.id=p(x)$}
\State status(x):=inactive
\If{$s(x)\neq null $}
\State send message(id(x),deactivate) to s(x)               \Comment{forward the deactivation message to successor}
\EndIf
\Else
\State send message(id(x),delete-successor) to m.id
\EndIf
\EndIf
\State
\newpage
\If{$m.type=activate$}
\If{$m.id=p(x)$}
\State status(x):=active
\If{$s(x)\neq null $}
\State send message(id(x),activate) to s(x)                 \Comment{forward the activation message to successor}
\EndIf
\Else
\State send message(id(x),delete-successor) to m.id
\EndIf
\EndIf
\State

\If{$m.type=forward-from-successor$}                  \Comment{insert the node forwarded from $s(x)$ to $N(x)$}
\If{$m.id=s(x)$}
\State insert(N(x), m.id2, head)
\EndIf
\EndIf
\State

\If{$m.type=forward-from-predecessor$}                    \Comment{insert the node forwarded from $p(x)$ to $N(x),L(x)$}
\If{$m.id=p(x)$}
\State insert(N(x), m.id2, tail)
\State insert(L(x), m.id2, head)
\Else
\State send message(id(x),delete-successor) to m.id
\EndIf
\EndIf
\State

\end{algorithmic}
\end{algorithm*}

\section{Correctness}

In this section we show the correctness of our approach for the self-stabilizing clique.

At first we show some basic lemmas. We then show that in linear time all nodes belong to the same heap. Then we show that the head of this heap (node with the maximal id) is connected with every node and vice versa after an additional time of $\mathcal O(n)$. From this state it takes $\mathcal O(n)$ more time until every node is connected to every other node and the clique is formed. We give a formal definition of the legal state.

\begin{definition}
Let $G$ be a network with node set $V$ and $max=\max \left\{v\in V\right\}$  be the node with the maximum id. Then G is in a {\em legal state} iff $\forall v \in V: N(v)=V-\left\{v\right\}$ and $\forall v\in V-\left\{max\right\}: p(v)>v$ and $\forall v\in V-\left\{max\right\}: s(p(v))=v$.
\end{definition}

Note that the legal state contains the clique and also a sorted list over the nodes. In this section we will prove the following theorem.

\begin{theorem}\label{theorem:main}
After $\mathcal O(n)$ rounds the network stabilizes to a legal state.
\end{theorem}

\subsection{Phase 0: Recovery to a valid state}

In this phase we show that the network can recover if the internal variables $p(x)$ and $s(x)$ are undefined or set to invalid values, e.g $p(x)<x$. We therefore define a state as valid state, if the nodes in a connected component given by $p(x)$ links form a tree and the successor's predecessor has to be the node itself. 

\begin{definition}
We say that the network  $G$ is in a {\em valid} state if $p(x)>x$ and
$s(x)<x$ for all $x\in V$ whenever $p(x)$ and $s(x)$ are defined and if
$y=s(x)$, then $x=p(y)$.
\end{definition}

\begin{theorem}\label{theorem:valid}
It takes at most 2 rounds until the network is in a valid state.
\end{theorem}

\begin{proof}
The network may be at an invalid state at the first round we consider.
That means that the variables $p(x),s(x)$ can have invalid values.
So $x$ could have set a node $u$ as its predecessor (i.e. $u=p(x)$) with $u<x$, which is not valid according to our protocol. Despite the presence of this invalid state, our protocol can recover from it very fast, so that the actual stabilization procedure can start. So if a variable is set faulty, that is $p(x)<x$ or $s(x)>x$, it will be set to $null$ after the first round, once the periodic actions will have been executed, as it is tested in the actions $checkifhead$ and $forwardtosuc$ if $p(x)<x$ and if $s(x)>x$.
Once each node has computed a valid predecessor, it will request it to accept it as a successor. So after the next round each node will (if possible) also have a valid successor, 
Moreover if $x$ notices that it is contacted from multiple nodes that think that $x$ has stored them as successors, $x$ contacts all these nodes but one (its true successor) through \emph{delete-successor} messages so at next round $x$ has no multiple successors. In other words it always holds that if $y=s(x)$, then $x=p(y)$.
\end{proof}

For our further analysis we assume that the initial state is valid, since we do not take into account the first
2 rounds it takes to reach a valid state. So we consider the first round in which we have a valid state as the round $t=0$. Note that due to the periodic actions the network stays in a valid state in every round afterwards.

\subsection{Phase 1: Connect all heaps by s-edges}

In this phase we show that starting from a valid state all existing heaps will eventually be connected by s-edges (defined below), so that they will merge afterwards.

First we give following definitions.

\begin{definition}
We distinguish between two different kinds of edges that can exist at any time in our network, the edges in the set $E$ and the ones in the set $E_s$.
We say that $(x,y)$ is in  $E$, if $y\in N(x)$ and $(x,y)$ in $E_s$ if $y\in S(x)$, resulting from a scan from $y$. We will call the latter ones \emph{s-edges} and denote them by $(x,y)_s$.
\end{definition}

\begin{definition}
In the directed graph we define an \emph{undirected path} as a sequence of edges $(v_0,v_1),(v_1,v_2),$ \newline $\cdots , (v_{k-1},v_k$), such that $\forall i \in \left\{1,\cdots ,k\right\}:(v_i,v_{i-1})\in E \vee (v_{i-1},v_i)\in E$. 
\end{definition}

\begin{definition}
We say that two heaps $H_1$ and $H_2$ are \emph{s-connected} if there exists at least one undirected path from one node in $H_1$ to one node in $H_2$ and this path consists of either s-edges or edges having both nodes in the same heap.
\end{definition}

\begin{definition}
We say that a subset of s-edges $E'_s \subseteq E_s$ is a \emph{s-connectivity set} at round $t$ if all heaps in the graph are s-connected to each other through edges in $E'_s$ at round $t$.
\end{definition}

In the first phase we will show that after $\mathcal O(n)$ rounds all heaps have been connected by s-edges. Let $E^0$ be the set of edges $(u,v)\in E$ at time $t=0$. We then show that all these edges are scanned in $\mathcal O(n)$ rounds, giving us the connections via s-edges.

\begin{theorem}\label{theorem:heap-connectivity}
After $\mathcal O(n)$ rounds the heaps $H_i$ and $H_j$ connected by $(u,v)\in E^0$ have either merged or been connected by s-edges .
\end{theorem}

To prove the theorem we firstly show some basic lemmas needed in the analysis.

\begin{lemma}\label{lem:stab-1}
Let $u_1,\cdots u_{|H|}$ be the elements in a heap $H$ in descending order. Then it takes at most $i$ rounds till $H$ is linearized w.r.t $u_i$.
\end{lemma}

\begin{proof}
We prove the lemma by induction on the number of rounds $i$. Note that all nodes are connected by the $p(x)$ links only to nodes with larger ids.

Induction base ($i=0$): The head of the heap is the node with the maximal id therefore trivially, $\forall v\in H-\left\{h\right\}: p(v)>v$ and $\forall v\in H-\left\{h\right\}$ with $v\geq h: s(p(v))=v$.

Induction step ($i\rightarrow$ i+1): By induction the heap is linearized w.r.t. $u_i$ after $i$ rounds, thus $u_{i+1}$ has to be connected to $u_i$ by a $p(x)$ link. In the $i+1$th round $u_i$ sends $new-predecessor$ messages to all other nodes with $p(x)=u_i$, such that $s(u_i)=u_{i+1}$ and $u_{i+1}$ becomes the only node with $p(x)=u_i$. Then $\forall v\in H-\left\{h\right\}: p(v)>v$ and $\forall v\in H-\left\{h\right\}$ with $v\geq u_{i+1}: s(p(v))=v$.
\end{proof}

\begin{lemma}\label{lem:merge}
Once one head learns about the existence of another head, two heaps are merged.
\end{lemma}

\begin{proof}
Let $h_i$ be the head of heap $H_i$. Also, let $h_j$ be the head of heap $H_j$ scanning $h_i$. There can be two cases.

\begin{itemize}

\item $h_i<h_j$: In this case, $h_i$ will no longer be a head once $h_j$ scans it and sends it own id.

\item $h_i>h_j$: In this case, $h_j$ will no longer be a head and will send an \emph{pred-request} to $h_i$.
\end{itemize}
\end{proof}

In case of a merging of two heaps $H_i$, $H_j$, the time it takes until the new heap $H$ is linearized w.r.t. a node $u$ can increase with respect to the linearization time of $u$ in the heap before the merging.

\begin{lemma}\label{lem:stab-merge}
If two heaps $H_i$ and $H_j$ merge to one heap $H$, the linearization time of a node $u\in H_i$ (resp. $u \in H_j$) can increase by at most $|H_j|$ (resp. $|H_i|$).
\end{lemma}

\begin{proof}
Without loss of generality let $u\in H_i$. By Lemma ~\ref{lem:stab-1} we know that the linearization time depends on the number of nodes with a larger id in the heap. The number of nodes with a larger id can increase by at most the size of the other heap $H_j$. Thus, also the linearization time can only increase by at most $|H_j|$.
\end{proof}

From Lemma ~\ref{lem:stab-1} and Lemma ~\ref{lem:stab-merge} we immediately get via an inductive argument:

\begin{corollary}\label{lem:stabilization}
For any heap $H$ of size $|H|$ in round $t$ it takes at most $|H|-t$ rounds until it forms a sorted list.
\end{corollary}

\begin{lemma}\label{lem:id-delay}
If a node sends an $id$ with a  \emph{forward-from-successor} message, the $id$ will not be delayed by other  \emph{forward-from-successor} messages on its way to the head.
\end{lemma}

\begin{proof}
Once a node sends a message to its predecessor through a \emph{forward-from-successor} message, the number of rounds it takes to reach the head of its heap depends only on the path to the head and the linearization steps.
When a node $x$ receives a \emph{forward-from-successor} message, it stores the received $id$ at the head of its neighborhood list $N(x)$. So this $id$ will be forwarded immediately, if $x$ is active. If it cannot be forwarded because $x$ is inactive, the node will inform its successor about its inactive state and as a consequence no more \emph{forward-from-successor} messages will be sent to $x$. That means that no other id can take the place of the one present at the head of $N(x)$.
So, once $x$ is active again, the $id$ will be sent immediately.
\end{proof}

As a consequence of the observation of Lemma ~\ref{lem:merge} we introduce some additional notation to estimate the time it takes until any id is scanned by a head of a heap.

For any edge $(u,v)\in E^0$ with $u\in H_i$ and $v\in H_j$, where $h_i$ and $h_j$ denote the corresponding heads of the heaps, we define the following notation in a round $t$: Let $P^t(u)$ be the length of the path from $u$ to $h_i$, once $H_i$ is linearized w.r.t. $u$. Let $ID^t(u,v)$ be the number of ids $u$ forwards or scans before sending or scanning $v$ the first time. Let $LT^t(u)$ be the time it takes until the heap is linearized w.r.t. $u$ , i.e. on the path from the head $h_i$ to $u$ each node has exactly one predecessor and successor. Corollary ~\ref{lem:stabilization} shows that $LT^t(u)$ is bounded by $|H_i|$.

Let $\phi^t(u,v)= P^t(u)+ID^t(u,v)+LT^t(u)$. We call $\phi^t(u,v)
$ the \emph{delivery time} of an id $v$ because if $\phi^t(u,v)=0$, the id is scanned in round $t$ or has already been scanned by $h_i$. We then denote by $\Phi^t(u,v)=\min\left\{\phi^t(w,v):heap(u)=heap(w)\right\}$ the minimal delivery time of $v$ for any node in the same heap as $u$.

For any edge $(u,v)\in E^0$, with $u\in H_i$ and $v\in H_j$, (i.e. $u$ and $v$ are in different heaps) and $\Phi^t(u,v)=0$ the head of $H_i$ scans or has scanned $v\in H_j$ resulting in the s-edge $(v,h_i)_s$.The following holds:

\begin{lemma}\label{lem:phi}
If $(u,v)\in E^0$ is an edge between two heaps $H_i$ and $H_j$, then
$\Phi^t(u,v)\leq \max \left\{2|H_i|+n-t,0\right\} $ \newline $\leq \max \left\{3n-t,0\right\}$ for all rounds $t$.
\end{lemma}

\begin{proof}
We will show the lemma by induction on the number of rounds. For the analysis we divide each round $t\rightarrow t+1$ into two parts: in the first step $t \rightarrow t'$ all actions are executed and in the second step $t'\rightarrow t+1$ all network changes are considered.
 Thus, we assume that all actions are performed before the network changes. This is reasonable as a node is aware of changes in its neighborhood only in the next round, when receiving the messages.
By network changes we mean the new edges that could be created in the network. These new edges could possibly lead to the merging of some heaps at time $t+1$.
\newline

Induction base($t=0$):

For any edge $(u,v)\in E^0$ between $H_i$ and $H_j$ let $x\in H_i$ be the node such that $\Phi^0(u,v)=\phi^0(x,v)$. Then $P^0(x)\leq H_i$ as the path length is limited by the number of nodes in the heap, $ID^0(x,v)\leq n$ as not more than $n$ ids are in the system, and following from Lemma ~\ref{lem:stabilization}, $LT(x)\leq |H_i|$. Then $\Phi^0(u,v)\leq\phi^0(x,v)\leq 2|H_i|+n\leq 3n$.
\newline

Induction step($t\rightarrow t'$):
For any edge $(u,v)\in E^0$ between $H_i$ and $H_j$ let $x\in H_i$ be the node such that $\Phi^t(u,v)=\phi^t(x,v)$.

Then in round $t$ the following actions can be executed.
\begin{itemize}
    \item $x$ is inactive and can not forward an id. Then the heap is not linearized w.r.t. $x$, which implies that the linearization time decreases by one, i.e. $LT^{t'}(x)=LT^t(x)-1$ and $\phi^{t'}(x,v)=\phi^{t}(x,v)-1 \leq 2|H_i|+n-t-1$ as all other values are not affected.
    \item $u$ is active, but does not send $v$ by a \emph{forward-from-successor} message, then the number of ids that $u$ is sending before $v$ decreases by 1. Note that according to Lemma ~\ref{lem:id-delay}, $x$ hasn't sent a \emph{forward-from-successor} message with $v$ in a round before, as then there would be another node $y\in H_i$ with $\phi^t(y,v)<\phi^t(x,v)$. Then $ID^{t'}(x,v)\leq ID^t(x,v)-1$ and $\phi^{t'}(x,v)) =\phi^{t}(x,v)-1\leq 2|H_i|+n-t-1$.
    \item $u$ sends a \emph{forward-from-successor} message with $v$, then the length of the path for $v$ to the head $h_i$ decreases by 1 and $\phi^{t+1}(p(x),v)\leq P^t(x)-1+ID^t(x,v)+LT^t(x) =\phi^{t}(x,v)-1\leq 2|H_i|+n-t-1$
\end{itemize}
Thus, in total $\Phi^{t'}(u,v)\leq \Phi^{t}(u,v)-1\leq 2|H_i|+n-t-1\leq 3n-(t+1)$.
\newline

Induction step($t'\rightarrow t+1$):
Now we consider the possible network changes and their effects on the potential $\Phi^{t+1}(u,v)$. Let again $x\in H_i$ be the node such that $\Phi^t(u,v)=\phi^t(x,v)$ for an edge $(u,v)\in E^0$ between $H_i$ and $H_j$. The following network changes might occur:

\begin{itemize}
    \item some heaps $H_k$ and $H_l$ with $k\neq i$ and $l\neq i$ merge. This has no effect on $\Phi^{t'}(u,v)$. Thus, $\Phi^{t+1}(u,v)=\Phi^{t'}(u,v) \leq 2|H_i|+n-t-1\leq 3n-(t+1)$.
    \item Heaps $H_i$ and $H_k$ merge to $H'_i$.  Obviously the length of the path of $x$ can increase and $P^{t+1}(x)\leq P^{t'}(x)+|H_k|$. According to Lemma ~\ref{lem:stab-merge} also the linearization time of $x$ can increase and $LT^{t+1}(x)\leq LT^{t'}(x)+|H_k|$. In total $\Phi^{t+1}(u,v)\leq \Phi^{t'}(u,v)+2|H_k|\leq 2|H'_i|+n-t-1\leq 3n-(t+1)$.
\end{itemize}

Thus, in round $t+1$, $\Phi^{t+1}(u,v)\leq 2|H_i|+n-t-1\leq 3n-(t+1)$.
\end{proof}

Hence for every edge $(u,v)\in E^0$ with $u\in H_i$ and $v\in H_j$, $\Phi^t(u,v)=0$ after $3n$ rounds, which means that the head of $H_i$ scans or has scanned $v\in H_j$ resulting in the s-edge $(v,h_i)$. Thus, we immediately get Theorem ~\ref{theorem:heap-connectivity}.

\subsection{Phase 2: Towards one heap}

Based on the results of Phase 1, we will prove that after $O(n)$ further rounds a clique is formed.
For the purpose of the analysis below, we use the following definitions:

\begin{definition}
Let $ord(x)$ be the \emph{order} of a node $x$, i.e. the ranking of the node if we sort all $n$ nodes in the network according to their id ( i.e. the node with the largest id $m$ has $ord(m)=0$, the second largest has order 1, and so on).
\end{definition}

\begin{definition}
We define the potential $\lambda(x,y)$ of a pair of nodes $x$ and $y$ to be the positive integer equal to $\omega(x,y)=2\cdot ord(x)+2 \cdot ord(y)+K(x,y)$, where $K(x,y)=1$ if $x>y$ and 0 otherwise.
Also, let  for a set of edges $E' \subseteq E$,  $\Lambda(E')=\max_{(u,v) \in E'}\{\omega(u,v)\}$, if $E' \neq \emptyset$ and 0 otherwise.
\end{definition}

We proceed by showing the following lemma.

\begin{lemma}\label{lem:s-edge connectivity}
Two heaps $H_i$, $H_j$ that are connected by an s-edge $(x,y)_s$ at time $t$ will either stay connected via s-edges $(x_i,y_i)_s$ at time $t+1$ with the property that, $\forall (x_i,y_i)$, the potential $\omega(x_i,y_i)$ of the edges we consider at time $t+1$ is smaller that the potential $\omega(x,y)$ of the edge $(x,y)_s$ we considered at time $t$,
or $x$ and $y$ will be in the same heap.
\end{lemma}

\begin{proof}
Let $(x,y)_s$ be a s-edge connecting $H_i$ and $H_j$, i.e. $x\in H_i$, $y\in H_j$. Then according to our algorithm the following actions might be executed.

\begin{itemize}
    \item $x$ is the head of $H_i$ and $y>x$ then $y=p(x)$ and $x$ sends a \emph{pred-request} message to $y$, resulting in a merge of $H_i$ and $H_j$.
    \item $x$ is the head of $H_i$ and $x>y$ and $y$ is a new id, then $x$ sends a \emph{scan-ack} to $y$ with its own id and the edge $(y,x)_s$ is created connecting $H_i$ and $H_j$. Then $\omega(y,x)= 2ord(x)+2ord(y)+0 < 2ord(x)+2ord(y)+1 =\omega(x,y)$.

    \item $x$ forwards $y$ to $p(x)$ by a \emph{forward-head} message, such that $y\in S(p(x))$ and $H_i$ and $H_j$ are connected by $(p(x),y)_s$. Then
$\omega(p(x),y)= 2 ord(p(x))+2 ord(y)+K(p(x),y)< 2 ord(x)+ 2 ord(y)+K(x,y) =\omega(x,y).$
    \item $x$ receives a new id $z\in S(x)$ with $z=\max\left\{v\in N(x)\right\}$, such that $z>y$ and $z>x$. Then $x$ sends a \emph{scan-ack} containing $z$ to $y$ and the s-edge $(x,y)_s$ is substituted by s-edges $(x,z)_s$ and $(y,z)_s$. And $H_i$ and $H_j$ are connected via s-edges.
Note that since $p(x)>x$ and  $z>x,y$ , $ord(p(x))<ord(x),ord(z)<ord(x)$ and $ord(z)<ord(y)$.
The potential of the new edges is:
$\omega(p(x),z)= 2 ord(p(x))+2 ord(z)+K(p(x),z)< 2 ord(x)+ 2 ord(y)+K(x,y) =\omega_t(x,y).$
$\omega(y,z)= 2 ord(y)+2 ord(z)+0< 2 ord(x)+ 2 ord(y)+K(x,y) =\omega_t(x,y).$
    \item $x$ knows an id $z\in H_k$ with $z=\max\left\{v\in N(x)\right\},z>y$ and $z\notin S(x)$. Then one of the following cases hold:
\begin{enumerate}
    \item  $(x,z)\in E^0$, then according to Lemma ~\ref{lem:phi} a node $u>x$ with $u\in H_i$ has scanned $z$ resulting in the s-edge $(z,u)_s$ s-connecting $H_i$ and $H_k$.
    \item $x$ has received $z$ by a \emph{forward-from-predecessor} message. Then a node $u>x$ with $u\in H_i$ has scanned $z$ resulting in the s-edge $(z,u)_s$ s-connecting $H_i$ and $H_k$.
        \item $z$ was in $S(x)$ in a previous round, then the edge $(x,z)_s$ existed s-connecting $H_i$ and $H_k$.
    \item $x$ has received $z$ by a \emph{forward-from-successor} message. Then there is a node $v\leq x$ in the sub heap rooted at $x$ such that $(v,z)\in E^0$. Then according to Lemma ~\ref{lem:phi} a node $w \in H_i$ with $w>v$ has scanned $z$ and the s-edge $(z,w)_s$ existed s-connecting $H_i$ and $H_k$. If $w>x$,  $H_i$ and $H_k$ are s-connected by s-edges $(x_i,y_i)_s$ with $\forall (x_i,y_i): (x<w<x_i \wedge x<w<y_i \wedge  z\leq x_i \wedge z\leq y_i) \vee (x<w\leq x_i \wedge x<w\leq y_i \wedge   z< x_i \wedge z< y_i)$. If $w<x$ then at least as many rounds have passed since $w$ has scanned $z$ as there are nodes on the path from $w$ to $x$, because $z$ has to be forwarded as many times. Then the edge $(z,w)_s$ has been forwarded or substituted $t$ times or $H_i$ and $H_k$ have merged. Then $H_i$ and $H_k$ are s-connected by s-edges $(x_i,y_i)_s$ with $\forall (x_i,y_i): (x<w<x_i \wedge x<w<y_i \wedge  z\leq x_i \wedge z\leq y_i) \vee (x<w\leq x_i \wedge x<w\leq y_i \wedge   z< x_i \wedge z< y_i)$.
\end{enumerate}
In each case $x$ sends a \emph{scan-ack} containing $z$ to $y$ and the s-edge $(y,z)_s$ is created. And $H_i$ and $H_j$ are s-connected over s-edges and in all cases the potential shrinks, since for each new s-edge it holds that at least one node is greater and the other node not smaller than the nodes in the edge they replace.
    \item $x$ is the head of $H_i$ and $x<y$, then $H_i$ and $H_j$ merge to one heap.
    \item $x$ is the head of $H_i$ and $x>y$ and $y$ was in $N(x)$ in a previous round, then $H_i$ and $H_j$ are already s-connected by s-edges $(x_i,y_i)_s$ with greater ids by the same arguments as in the case before. Since the ids are greater, the potential shrinks also here.
\end{itemize}
\end{proof}

\begin{lemma}\label{lem:scon}
If $E_t$ is an s-connectivity set at round $t$, there exists an s-connectivity set $E_{t+1}$ at round $t+1$ such that $\Lambda(E_{t+1})<\Lambda(E_t)$.
\end{lemma}

\begin{proof}
Let $E_t$ be an s-connectivity set a round $t$. 
 We replace every edge $(x,y)_s\in E_t$ with the edges $(x_i,y_i)_s$ as described in the lemma above.
For every pair of heaps that were s-connected at $t$ through an edge in $E_t$, there exists a set of s-edges of smaller potential that s-connects the two heaps at $t+1$. We include these edges in $E_{t+1}$. But at round $t$ all pairs of heaps are s-connected through $E_t$, which means that at round $t+1$ all pairs of heaps are also s-connected through $E_{t+1}$. So, $E_{t+1}$ is an s-connectivity set at round $t+1$.
 Also since all the edges in $E_{t+1}$ have less potential as the ones the replaced in $E_t$, $\Lambda(E_{t+1})<\Lambda(E_t)$.
\end{proof}

\begin{theorem}\label{theorem:heap}
After at most 4n+1 rounds, all heaps have been merged into one.
\end{theorem}

\begin{proof}
From Theorem ~\ref{theorem:heap-connectivity} we know that all heaps are s-connected after $\mathcal O(n)$ rounds. So after $\mathcal O(n)$ rounds there exists the first s-connectivity set, $E_0$, with $\Lambda(E_0)=\max_{(u,v) \in E_0}\{\omega(u,v)\}=\max_{(u,v) \in E_0}\{2ord(u)+2ord(v)+K(u,v)\}\leq2n+2n+1=4n+1$.
Since for each round $t$ and an s-connectivity set $E_t$, an s-connectivity set $E_{t+1}$ for round $t+1$ can be found, such that $\Lambda(E_{t+1})<\Lambda(E_t)$, (i.e. the potential of the s-connectivity set shrinks by every round)
after at most $4n+1$ rounds (after the existence of $E_s$) there exists an s-connectivity set $E_{\infty}$, such that $\Lambda(E_{\infty})= 0$. This means that $E_{\infty}$ is the empty set.
Since $E_{\infty}$ is an empty s-connectivity set connecting all the heaps of the graph,
we know that the graph has only one heap.
\end{proof}

\subsection{Phase 3: Sorted list and Clique}
\begin{theorem}\label{theorem:clique}
If all nodes form one heap, it takes $\mathcal O(n)$ time until the network reaches a legal state.
\end{theorem}

\begin{proof}
Since at this point we only have one head the heap will be linearized after $\mathcal O(n)$ rounds. This follows directly from Lemma ~\ref{lem:stabilization}.
Once the heap is linearized and forms a sorted list, each node's $id$ will be sent to the root, the remaining head, after at most $n$ rounds. So the root will be aware of every $id$. The root, as it sends according to the round-robin process all its information to its successor, will send after $n$ rounds all the $id$s to it, and the successor will do the same. As a consequence, all nodes will receive all $id$s at $O(n)$ rounds. Adding all this together, after $O(n)$ all nodes will know each other and a clique will be constructed.
\end{proof}

Combining Theorem ~\ref{theorem:valid}, Theorem ~\ref{theorem:heap-connectivity}, Theorem ~\ref{theorem:heap} and Theorem ~\ref{theorem:clique} our main theorem Theorem ~\ref{theorem:main} holds.

\section{Message complexity}
In this section we give an upper bound for the work spent by each node. We already mentioned that we will distinguish two types of work. The stabilization work, that is spent until a clique is formed, and the maintenance work, that is spent in each round in a legal state. We count the work of a node in the number of messages sent and received.

\subsection{Stabilization work}
According to Theorem ~\ref{theorem:main} it takes $\mathcal O(n)$ rounds to reach a legal state. In each round each active node sends a message to its predecessor and its successor (\emph{forward-from-successor}, \emph{forward-from-predecessor})  and receives a message from them (\emph{forward-from-successor}, \emph{forward-from-predecessor}).
Also, a node sends at most one \emph{activate/deactivate} message to its successor at each round.
This gives a resulting work of $\mathcal O(n)$ for each node or $\mathcal O(n^2)$ in total. By the following lemmas we show that the additional messages sent and received during the linearization are at most $\mathcal O(n)$ for each node.

\begin{lemma}\label{lem:stab1}
Each node sends and receives at most $\mathcal O(n)$ \emph{pred-request}, \emph{pred-accept} and \emph{new-predecessor} messages during the linearization phase.
\end{lemma}

\begin{proof}
In each round each node sends at most one \emph{pred-request} and one \emph{pred-accept} message and receives at most one \emph{pred-accept} or \emph{new-predecessor} message. It remains to show that each node receives at most $\mathcal O(n)$ \emph{pred-request} and sends at most $\mathcal O(n)$ \emph{new-predecessor} messages. Note that it suffices to show that each node receives at most $\mathcal O(n)$ \emph{pred-request}, as the number of \emph{new-predecessor} messages directly depends on the number of received \emph{pred-request} messages, to each node, that sends a \emph{pred-request} to $u$ that is not $u$' successor, $u$ sends a \emph{new-predecessor} message. A node $u$ only sends at most one \emph{new-predecessor} message to each other node $v$. By receiving this message $v$ changes its predecessor. Thus before $u$ sends another \emph{new-predecessor} message to $v$, $v$ has to change its predecessor back to $u$. A predecessor is only changed if a root receives an id greater than its own id, or if the predecessor of a node sends a \emph{new-predecessor}. $v$ cannot be a head, thus $v$'s predecessor is only changed by another \emph{new-predecessor} message. But $v$'s predecessor can not be changed back to $u$ as the id of the new predecessor is strictly decreasing. By this monotonicity it follows that a node $u$ only sends at most one \emph{new-predecessor} message to each other node $v$. Thus, every node only sends and receives $\mathcal O(n)$ \emph{pred-request} and new predecessor messages.
\end{proof}

\begin{lemma}\label{lem:stab2}
Each node sends and receives at most $\mathcal O(n)$ \emph{scan} and \emph{scan-ack} messages during the linearization phase.
\end{lemma}

\begin{proof}
Only heads of heaps send scan messages. In each round each head sends exactly one \emph{scan} message. Each scanned node sends a \emph{scanack} message back or stores the id of the head in $S(x)$. Obviously a node can be scanned by up to $n$ different heads in one round. Which would lead to a work of $\mathcal O(n^2)$ by receiving these messages. But as a node sends the maximal id in its neighborhood with a \emph{scanack} message, it is scanned at most once by heads with an id smaller then $max$. By receiving this id the scanning node recognizes, if it is still a head, that it is not the largest id and cannot be a head of the heap and sets its predecessor and stops scanning. So a node can be scanned by $\mathcal O(n)$ heads before the heads stop scanning, because they received a \emph{scanack}. A head that is not the maximal head, that scanned the node so far, will only scan the node one more time and then stop scanning. So a node receives at most $\mathcal O(n)$ scan messages from a new maximal head, $\mathcal O(n)$ messages from the current maximal head, as each head only sends one scan message per round, and all other scans increase the number of inactive heads, which is limited by $\mathcal O(n)$. Regarding the \emph{scanack} messages, since each head scans only once in each round, it receives also at most one \emph{scanack} (that result from sent \emph{scan} messages) message in each round. A node $x$ can also receive a  \emph{scanack} message when sending a \emph{scanack} message, but this happen only the if the node to which the \emph{scanack} was sent does not know $x$, so all in all at most  $n$ times.  So, all in all, a node receives $\mathcal O(n)$ at the whole linearization phase.
\end{proof}

\begin{lemma}\label{lem:stab3}
Each node sends and receives at most $\mathcal O(n)$ \emph{forward-head} messages through the linearization phase.
\end{lemma}

\begin{proof}
Moreover, a node sends at most one \emph{forward-head} message per round.
The number of \emph{forward-head} messages it receives during the linearization phase is limited by $\mathcal O(n)$. That is because each node $x$ receives one \emph{forward-head} message from its successor in a round, and possibly from other possible successors, let $u$ be such one, for which $p(u)=x$. But $u$ can only be once a possible successor of $x$, since at the next round it either will be forwarded to $s(x)$ and will never have $x$ as its predecessor again, or it becomes $s(x)$.
Since each node can be only once a possible successor for $x$, the number of \emph{forward-head} messages sent through all possible successors is limited by $n$. So, the number of \emph{forward-head} messages it receives during the linearization phase is limited by $\mathcal O(n)$.
\end{proof}

\subsection{Maintenance work}

\begin{lemma}\label{lem:maint}
As soon as the network forms a stable clique with  a stable list as a spanning tree, i.e. the network is in a legal state, each node sends and receives at most $\mathcal O(1)$ messages in each round.
\end{lemma}

\begin{proof}
In a legal state all nodes form a sorted list. Thus, each node has exactly one stable successor and one stable predecessor. Then each node sends and receives one \emph{pred-request} and one \emph{pred-accept} message. Each node sends one \emph{forward-from-successor} and one \emph{forward-from-predecessor} message. Moreover there is one head that sends one \emph{scan} message, which is received by one other node, and receives one \emph{scanack}, sent by the scanned node. Thus, each node sends and receives $\mathcal O(1)$ messages in a stable state.
\end{proof}

\section{Single Join and Leave Event}
The case of arbitrary churn is hard to analyze formally. Thus, we will show that the clique can efficiently recover considering a single join or leave event in a legal state.

\begin{theorem}\label{theorem:join}
In a legal state it takes $\mathcal O(n)$ rounds and messages to recover and stabilize after a new node joins the network. It takes $\mathcal O(1)$ rounds and messages to recover the clique after a node leaves the network.
\end{theorem}

\begin{proof}
If a node $u$ joins the network it creates an edge $(u,v)$ to a node $v$ in the clique. If $v>u$, $u$ sends a \emph{pred-request} to $v$, $v$ then either accepts $u$ as its successor or creates an edge from $u$ to $v$'s successor. It takes at most $\mathcal O(n)$ rounds until $v$ reaches its final position in the sorted list. Additionally $v$ sends $u$'s id to its predecessor, and after $\mathcal O(n)$ rounds the head inserts $u$ to its neighborhood.
If $v<u$ $v$ sends $u$'s id to its predecessor, because it is a new id. Then it takes at most $\mathcal O(n)$ rounds until the head receives $u$'s id and scans $u$, then $u$ assumes the head to be its predecessor and case 1 holds.
After $\mathcal O(n)$ further rounds each nodes receives $u$'s id and $u$ receives the id of all other nodes in the network. Thus, after $\mathcal O(n)$ rounds after a join the nodes form a clique and the sorted list is linearized.

Obviously a clique remains a clique in case a node $u$ leaves the network. Also the sorted list is immediately repaired, as the successor of the removed node, assumes $u$' predecessor to be its predecessor and sends a \emph{pred-request}, which will be accepted as the node has no other successor. Note that if $u$ is the head of the list, $u$'s successor will recognize that there is no node with a larger id in its neighborhood and will correctly assume to be a head of a list and proceed the scanning.
\end{proof}

\section{Conclusion}
In this paper we introduced a local self-stabilizing time-and work-efficient algorithm that forms a clique out of any weakly connected graph. By forming a clique our algorithm also solves the resource discovery problem, as each node is aware of any other node in the network. Our algorithm is the first algorithm that solves resource discovery in optimal message complexity. Furthermore our algorithm is self-stabilizing and thus can handle deletions of edges and joining or leaving nodes.


\bibliographystyle{plain}


\pagebreak
\begin{appendix}

\end{appendix}

\end{document}